\documentclass[12pt,oneside,english]{amsart}
\usepackage[utf8]{inputenc}
\usepackage[a4paper]{geometry}
\usepackage{footnote}
\usepackage{color}
\usepackage{babel}
\usepackage{array}
\usepackage{subfig}
\usepackage{url}
\usepackage{amssymb}
\usepackage{verbatim}
\usepackage{hyperref}
\usepackage[foot]{amsaddr}
\usepackage{enumitem}
\usepackage{multirow}
\usepackage{amsmath}
\usepackage{amsthm}
\usepackage{graphicx}
\usepackage{setspace}
\usepackage[authoryear]{natbib}
\usepackage{CJKutf8}

\makeatletter

  \theoremstyle{plain}
  \newtheorem{prop}{\protect\propositionname}
   
  \theoremstyle{plain}
  \newtheorem{cor}{\protect\corollaryname}
  \theoremstyle{remark}
  
  \theoremstyle{definition}
  
\theoremstyle{plain}

  \theoremstyle{plain}
  \newtheorem{lem}{\protect\lemmaname}

\definecolor{cite-blue}{RGB}{0,0,204}
\date{}

\makeatother

\usepackage[all]{xy}
\bibpunct{(}{)}{;}{a}{,}{,}
\usepackage{tikz}
\usetikzlibrary{patterns}
\usetikzlibrary{datavisualization.formats.functions}
\usetikzlibrary{fit}
\usetikzlibrary{fadings}
\usetikzlibrary[graphs]
\usetikzlibrary{shapes.misc, positioning}

  \providecommand{\definitionname}{Definition}
  
  \providecommand{\lemmaname}{Lemma}
  \providecommand{\propositionname}{Proposition}
  \providecommand{\remarkname}{Remark}
\providecommand{\corollaryname}{Corollary}
\providecommand{\theoremname}{Theorem}

\title{Incentive-compatible pricing of public transportation}

\author{Inácio Bó}
\thanks{\textbf{Inácio Bó}: University of Macau, Faculty of Social Sciences, Department of Economics, Macau; website: \protect\url{https://www.inaciobo.com/}; e-mail: inaciobo@um.edu.mo}

\author{Chiu Yu Ko}
\thanks{\textbf{Chiu Yu Ko}: Department of Decisions, Operations and Technology, The Chinese University of Hong Kong, Hong Kong; website: \protect\url{https://www.bschool.cuhk.edu.hk/staff/ko-chiu-yu/};
email: chiuyuko@cuhk.edu.hk.}
\thanks{Inácio Bó acknowledges financial support from the National Natural Science Foundation of China (Grant No. 72573203) and the Asia-Pacific Academy of Economics and Management seed grant.}
\date{\today}

\begin{document}

%\begin{abstract}
%    We study pricing of public transportation enforced by random inspection (proof-of-payment) rather than physical barriers. Strategic riders may evade payment, buy partial tickets, or split trips to minimize expected costs. We derive the highest incentive-compatible prices in both line and network models, showing that optimal pricing takes a simple additive form. Applied to the Washington DC Metro, switching to proof-of-payment with current fares (holding demand fixed and optimizing route choice) reduces revenue by over 60\%. Targeted fare reductions to restore incentive compatibility recover most of this, limiting revenue loss to 15–21\%.
%\end{abstract}

\begin{abstract}
  We study how to price public transportation when payment is enforced by random inspection rather than by physical barriers. Riders are strategic: they may travel uncovered, buy partial tickets, or split a trip into multiple tickets in order to reduce expected payment. We derive the highest prices that remain incentive-compatible, first in a line model and then in a network model, where riders might have multiple routes from origin to destination. We derive the highest incentive-compatible prices, showing that optimal pricing takes a simple additive form. We then apply the model to the Washington DC Metro. Holding demand fixed and solving for route choice over all feasible routes in the network, we show that switching to proof-of-payment while keeping current fares could reduce fare revenue by more than 60\%. By contrast, simply reducing some targeted fares to make them incentive compatible can recover most of that revenue loss, reducing it to between 15\% and 21\%.
\end{abstract}

\maketitle

\section{Introduction}

Many services charge users by controlling access ex ante: payment is made at the entrance and non-payers are physically excluded. Public transportation systems often work this way through turnstiles, gated stations, or front-door payments on buses. A proof-of-payment system takes a different approach. Riders are allowed to enter and travel freely, and compliance is enforced ex post through random inspection and fines. This design has clear operational advantages. It reduces infrastructure and staffing requirements, allows stations and vehicles to operate with fewer physical bottlenecks, and gives operators greater flexibility in how access points and stops are designed. The downside is equally clear: once physical barriers are removed, riders can underpay more easily.

The pricing problem is therefore different under proof-of-payment. A rider can evade in several ways. She may travel uncovered, buy a ticket for only part of the trip, or split a trip into multiple tickets if the fare schedule makes that profitable. Once these possibilities are admitted, ticket design cannot be separated from incentives. Prices that are reasonable under gated entry may induce substantial underpayment under random inspection.

This paper studies the highest prices that remain incentive-compatible in proof-of-payment systems. We first analyze a line model, which isolates the core logic and yields the pricing rule in its cleanest form. We then turn to a general network model, which shows how the same logic extends to realistic transit systems in which passengers may choose among several routes. Throughout, demand is treated as fixed. This is a deliberate modeling choice: our objective is not to solve a general-equilibrium transit-demand problem, but to characterize fares that deter underpayment conditional on a travel need. The resulting prices should therefore be interpreted as prices that preserve truthful payment among riders who still choose to travel, which is the natural object when the policy question is how to run proof-of-payment while maintaining service use.

Our main theoretical result is simple and sharp. For a given inspection technology, there is a maximal incentive-compatible price for each trip---the highest fare a rider will still pay rather than risk the fine from traveling uncovered. This price equals the expected fine exposure from traveling uncovered. Pricing works in the same way in the network model. However, this solution is not trivial in networks because the incentive-compatible prices depend on the equilibrium flow of passengers, which itself depends on the set of available routes, their prices, and the strategic choices of all riders. The existence of such a pricing system requires that, at equilibrium, no passenger has an incentive to deviate from the route and ticketing strategy assumed by the pricing structure. We show that such pricing schedule exists by formulating a game in which passengers choose routes and ticket purchases, but where prices are endogenous to these choices. This game has an equilibrium, which itself results in a pricing schedule of incentive-compatible prices that supports truthful payment and equilibrium route choices.

We then apply the model to the Washington DC Metro using network and origin-destination demand data. The application keeps demand fixed but allows route choice over all feasible routes for each origin-destination pair, and computes incentive-compatible prices at the resulting route-equilibrium flows. A switch to proof-of-payment while keeping current DC fares would produce a revenue loss of \textbf{more than 60\%}. By contrast, adjusting fares so that riders do not profit from underpayment reduces the revenue loss to \textbf{between 15\% and 21\%}, while making many trips cheaper than under the current system. Fare design that accounts for rider incentives can therefore have large revenue consequences in proof-of-payment systems. Incentive-compatible pricing addresses this, in this application, by adjusting only the fares that create profitable evasion, leaving the rest of the schedule in place. This can recover revenue that a price structure blind to incentives would sacrifice.

\section{Related literature}

Our paper relates first to the economics of evasion, monitoring, and punishment. The broader literature studies how compliance responds to inspection intensity, fines, and the strategic environment in settings such as shoplifting, piracy, tax evasion, and parking violations \citep{becker_crime_1968,yaniv_shoplifting_2009,perlman_reducing_2014,chellappa_managing_2005,slemrod_cheating_2007,fisman_corruption_2007}. In Becker's benchmark model, maximal fines are attractive because they economize on enforcement. Later work shows that this conclusion need not survive once one allows for richer constraints or avoidance responses \citep{polinsky_optimal_1979,polinsky_optimal_1984,polinsky_economic_2000,malik_avoidance_1990}. Our contribution is not to solve for an optimal fine, but to characterize the highest fares that remain compatible with truthful payment when fines and inspection exposure are taken as given.

Second, the paper contributes to the fare-evasion literature in public transportation. Early contributions document the basic economics of fare evasion in transit \citep{boyd_fare_1989,kooreman_fare_1993}. More recent work studies the determinants of fare evasion, rider behavior in proof-of-payment systems, and optimal inspection intensity \citep{barabino_what_2015,delbosc_why_2019,barabino_fare_2020,barabino_moving_2019}. Relative to that literature, our focus is on fare design rather than on rider profiling or inspection policy alone. We ask how the fare schedule itself should be chosen when riders can respond strategically to proof-of-payment enforcement.

Third, our paper is related to the operations-research and security-game literature on inspections and monitoring in networks \citep{yin_trusts_2012,jain_security_2013,tirachini_multimodal_2012}. This literature typically takes the pricing environment as given and studies where to allocate inspectors or how to solve the resulting network optimization problem. The closest paper is \cite{correa_fare_2017}, which also features strategic riders in a transit network. There are two key formal differences. First, they study a bilevel (Stackelberg) problem in which the network operator chooses where to deploy inspectors and riders best-respond; fares are either fixed exogenously or optimized jointly with inspections. We take the opposite perspective: the inspection technology is given and the object of interest is the maximal fare schedule compatible with truthful payment. Second, evasion in their model is binary---a rider either buys a full ticket or travels without one---so there is no scope for partial underpayment. In our framework, riders may buy a ticket covering only part of the trip or split a longer trip into cheaper segments, and incentive-compatible prices are derived precisely from ruling out these richer deviations. Given these modeling choices, the two papers yield complementary contributions: theirs is algorithmic, designing exact and approximation algorithms for the inspector-placement problem; ours is economic, providing a closed-form pricing rule and a quantitative assessment of its revenue implications. In addition, in our network model the inspection exposure on a segment depends on the passenger flow through that segment, which creates a direct link between network usage and incentive-compatible prices.

Finally, the application connects to the economics of transit pricing more broadly \citep{mohring_optimization_1972,keeler_optimal_1977,parry_how_2010,de_palma_economics_2017}. That literature studies efficient fares, congestion, service design, and pricing in urban transportation systems. Our paper is complementary. We abstract from endogenous demand and service choice in order to isolate a different object: the incentive-compatible fare schedule induced by proof-of-payment enforcement.

\section{Baseline model\label{section:baseline}}

There is a positive mass of nonatomic passengers residing on the open line $(0,1)$, meaning that the passenger population is modeled as a continuum and that each individual rider has measure zero. Denote by $\mathcal{I}$ the set of passengers.\footnote{The use of the open interval is only a convenience: trips start and end in the interior, and working on $(0,1)$ lets us avoid boundary cases at the endpoints without changing the substance of the model.} Passengers have travel demand density $d\left(x,y\right)$, where $x$ is the origin and $y$ the destination. The function $d$ is continuously differentiable, and $d\left(x,y\right)>0$ for any $x,y\in(0,1)$ with $x\neq y$. For any interior origin $x_0\in(0,1)$, the quantity $\int_0^1 d\left(x_0,y\right)dy$ is the mass of passengers who want to travel from $x_0$.

An important object for our purposes is the density of passengers who pass through a point $a$ on their trip. For each $a\in(0,1)$, define

\begin{equation}
    d^{pass}\left(a\right) = \int_0^a \int_a^1 d\left(x,y\right)dy\ dx + \int_a^1 \int_0^a d\left(x,y\right)dy\ dx.
\end{equation}

The first term is the mass of passengers traveling from the left of $a$ to the right of $a$, and the second term is the mass traveling in the opposite direction. We assume throughout that $d^{pass}(a)>0$ for every $a\in(0,1)$.

To reduce fare evasion, the transportation authority has a mass $\lambda^{total}$ of inspectors and distributes them along the line. Without loss of generality, we normalize this mass to unity, that is, $\lambda^{total}=1$, so that the inspector allocation can be represented by a density $\lambda\left(a\right)$.

A mass $\lambda(a)$ of inspectors at point $a$ generates inspections as a Poisson process with aggregate rate $\phi(\lambda(a))$ per unit length, where $\phi$ is monotonically increasing, differentiable, and concave; each inspection event selects one passenger uniformly at random from the $d^{pass}(a)$ passengers passing through~$a$. Concavity of $\phi$ captures diminishing returns from concentrating inspectors at one location. Since passengers' trips start and end in $(0,1)$, every feasible trip is contained in a compact subinterval of $(0,1)$. We also assume that, for every compact interval $I\subset(0,1)$, the function $a\mapsto \phi(\lambda(a))/d^{pass}(a)$ is integrable on $I$. This condition simply rules out infinite expected inspection exposure on any feasible trip.

There is a \textbf{pricing scheme} $p:[0,1]^2\to\mathbb{R}_+$ for every pair of origin and destination. Any ticket can be purchased at every point of departure, at a price $p\left(x,y\right)$, where $x$ is the origin and $y$ the destination.\footnote{This includes the possibility that tickets can be purchased via mobile phones, in the train, or that previously purchased tickets are ``validated'' for use at any time inside of the bus or subway.}

We assume that a passenger who demands a trip makes it entirely using public transportation. She is free, however, to choose any path between her origin and destination, and to buy any set of tickets, if any, for that trip. This is realistic even on a line. For example, a rider traveling from $x$ to $z$ through an intermediate point $y$ may buy one ticket $(x,z)$ or two tickets $(x,y)$ and $(y,z)$.\footnote{A useful real-world example comes from Berlin. In March 2026, a regular single ticket costs EUR~4.00 and is valid for two hours, whereas a short-trip ticket costs EUR~2.80 but is valid for only three subway stops or six tram or bus stops. A rider who wants to travel four subway stops, or seven tram or bus stops, would have to pay about 43\% more for the regular ticket just to cover that additional stop. Alternatively, she could buy the cheaper short-trip ticket and be fully covered on all but the last uncovered segment of the trip; if inspected before then, she is valid, and only an inspection after the short-trip validity ends would generate a fine.}  Hence, a passenger's decision consists of the combination of (i) a path from her origin to her destination and (ii) a set of tickets. Formally, a \textbf{travel plan} for passenger $i\in\mathcal{I}$, with origin $x_i$ and destination $y_i$, is a pair $s_i=(\sigma_i,\tau_i)$. A path is a finite ordered list $\sigma_i=\left(z_i^0,z_i^1,\ldots,z_i^m\right)$ with $z_i^0=x_i$ and $z_i^m=y_i$. It induces the ordered multiset of traveled segments
\[
\mathcal{S}(\sigma_i)=\left\{\left(z_i^{h-1},z_i^h\right):h=1,\ldots,m\right\}.
\]
The ticket choice $\tau_i$ is a finite multiset of origin-destination pairs. A ticket $(a,b)$ is an ordered pair with origin $a$ and destination $b$; it covers every point in the interval $[\min\{a,b\},\max\{a,b\}]$, so coverage depends on the physical stretch between origin and destination. A ticket $(a,b)$ covers a traveled segment $(c,d)$ if $[\min\{c,d\},\max\{c,d\}]\subseteq [\min\{a,b\},\max\{a,b\}]$. If a purchased ticket covers only part of a traveled segment, that ticket's endpoint can be inserted into the path as an additional breakpoint, splitting the segment there. This refinement does not affect total ticket payment or total expected inspection exposure, since both quantities are additive over adjacent sub-intervals: the cost of covering a segment equals the sum of costs over any sub-partition of it, and the same holds for inspection exposure. Therefore, without loss of generality, we may assume that every segment of the path is either fully covered by one of the purchased tickets or entirely uncovered. We denote by $\mathcal{U}(\sigma_i,\tau_i)\subseteq \mathcal{S}(\sigma_i)$ the multiset of uncovered traveled segments.

The expected number of inspections experienced when traveling between $x$ and $y$ is given by

\begin{equation}
\label{eqn:ProbBeingInspected}
    q\left(x,y,\lambda\right)=\left|\int_x^y \frac{\phi\left(\lambda\left(a\right)\right)}{d^{pass}\left(a\right)} da\right|.
\end{equation}

Equation~\eqref{eqn:ProbBeingInspected} follows directly from the inspection technology. A given passenger therefore faces inspections at per-unit-length rate $\phi(\lambda(a))/d^{pass}(a)$, and the total expected number of inspections on a trip from $x$ to $y$ is the integral in~\eqref{eqn:ProbBeingInspected}. Each inspection of an uncovered passenger independently costs a fine of~$\alpha$, so the expected fine from traveling uncovered between $x$ and $y$ is $\alpha\,q(x,y,\lambda)$. When $\phi(x)=x$, for instance, the per-unit inspection rate is simply the inspector-to-passenger ratio. Note that $q$ is symmetric: $q(x,y,\lambda)=q(y,x,\lambda)$. This follows from the absolute value in~\eqref{eqn:ProbBeingInspected} and, economically, from the fact that $d^{pass}(a)$ counts passengers passing through~$a$ in both directions. As a consequence, the incentive-compatible price satisfies $p^*(x,y)=p^*(y,x)$: although tickets are formally directional objects, the maximal incentive-compatible fare is the same in both directions. The function $q$ is also additive on adjacent intervals: if $x\le z\le y$, then
\[
q(x,y,\lambda)=q(x,z,\lambda)+q(z,y,\lambda).
\]

If a passenger is inspected at any point $a\in(0,1)$ and does not have a valid ticket at that point, she receives a fine of $\alpha$. Moreover, if she has a ticket with origin-destination $(x,y)$ but either $a<x$ or $a>y$, then she is also punished with the same fine.\footnote{Not differentiating between these cases for fine purposes is in line with the real-life practices that we are aware of.} While we later discuss how the transportation authority should distribute inspectors, we do not count fine revenue as part of the authority's objective. This is natural when fines accrue to a different budget line, or when the objective is fare revenue rather than enforcement revenue.

We assume that passengers have quasi-linear utilities, are risk-neutral, and maximize expected utility. Passenger $i\in\mathcal{I}$ derives utility $u_i\left(x,y\right)$ from making a trip from $x$ to $y$. If she chooses travel plan $s_i=(\sigma_i,\tau_i)$, her expected utility is

\begin{equation}
  u_i\left(x,y\right) - \alpha\sum_{(a,b)\in \mathcal{U}(\sigma_i,\tau_i)} q(a,b,\lambda) -\sum_{(a,b)\in \tau_i} p\left(a,b\right),
\end{equation}
where the first term is the utility from completing the trip, the second term is the expected fine on uncovered segments, and the third term is total ticket expenditure. The inspection probabilities enter through the expected-exposure term $q$.

Because passengers are nonatomic, a unilateral deviation by a single passenger does not affect aggregate passenger flow, and hence does not affect $d^{pass}$ or the exposure function $q$. Every rider therefore takes $q$ as given when choosing how to travel and what tickets to buy. A pricing scheme is \textbf{incentive-compatible} if, given those aggregate conditions, almost every passenger finds it optimal to buy the ticket corresponding to her origin and destination and to travel directly between them.\footnote{Here ``almost every" means all passengers except possibly a measure-zero set.}

For any pattern of ticket purchases, let $m(x,y)$ denote the mass of passengers who buy the ticket $(x,y)$. The \textbf{total revenue} of the transportation authority is then

\begin{equation}
  \pi(p,m)=\int_0^1 \int_0^1 p\left(x,y\right)m(x,y)\,dy\,dx.
\end{equation}

A pricing scheme $p^*$ is \textbf{maximal among incentive-compatible prices} if no other incentive-compatible pricing scheme yields weakly higher revenue for every origin-destination pair and strictly higher revenue for some pair. In the present environment, this notion coincides with maximizing fare revenue among incentive-compatible pricing schemes.

In the present model, the passenger's relevant outside option is to travel without buying any ticket. Hence, no incentive-compatible pricing scheme can leave a passenger with payoff below $u_i\left(x,y\right)-\alpha q(x,y,\lambda)$. We say that a pricing scheme $p^*$ is \textbf{full-surplus-extracting} if it is incentive-compatible and every passenger $i\in\mathcal{I}$ obtains exactly the payoff $u_i\left(x,y\right)-\alpha q(x,y,\lambda)$ under $p^*$. The ``surplus'' being extracted is not the travel utility $u_i(x,y)$ but the gap between the travel utility and the no-ticket outside option. Under $p^*$ the rider is indifferent between paying the fare and traveling uncovered; the authority captures all of the willingness to pay above what the rider could guarantee by evading.

Our first result identifies the highest full-ticket price compatible with truthful payment. In what follows, we say that a path is \textbf{weakly dominated} by another if the latter yields weakly lower total cost---ticket payments plus expected fines---regardless of the ticket coverage chosen; the domination is \textbf{strict} if the cost reduction is strict.

\begin{lem}
\label{lem:directRoute}
Fix any nonnegative pricing scheme and any inspection exposure function $q$. For a passenger traveling from $x$ to $y$ on a line, every path that backtracks is weakly dominated by the direct path that traverses the interval between $x$ and $y$ exactly once. The domination is strict whenever the original path traverses some subinterval more than once or goes outside the interval between origin and destination.
\end{lem}

\begin{proof}
Any backtracking path contains all points on the direct path, and in addition either traverses some subinterval more than once or visits points outside the interval between origin and destination. Since ticket prices are nonnegative and expected fines are nonnegative on every traveled segment, all extra travel weakly increases total cost and never changes the travel utility $u_i(x,y)$. Therefore the direct path weakly dominates any path with backtracking, and the domination is strict whenever the path contains any strictly extra travel.
\end{proof}

\begin{prop}\label{prop1}
The pricing scheme $p^*\left(x,y\right)=\alpha q(x,y,\lambda)$ is incentive-compatible, full-surplus-extracting, and maximal among incentive-compatible pricing schemes.
\end{prop}

\begin{proof}
Fix a passenger $i$ with origin $x$ and destination $y$, and suppose without loss of generality that $x<y$. By Lemma~\ref{lem:directRoute}, it is enough to consider the direct path from $x$ to $y$.

Under the full ticket $(x,y)$, the passenger's utility is
\[
U_i^{Full}(x,y)=u_i(x,y)-p^*(x,y)=u_i(x,y)-\alpha q(x,y,\lambda).
\]
If she buys no ticket, then her utility is
\[
U_i^{None}(x,y)=u_i(x,y)-\alpha q(x,y,\lambda).
\]

Now consider an arbitrary ticket choice along the direct path. Let
\[
x=t_0<t_1<\cdots<t_m=y
\]
be a partition that contains all ticket endpoints. Since the path has been refined at every ticket endpoint, each interval $(t_{j-1},t_j)$ is either covered by some ticket or uncovered. Let $C\subseteq\{1,\ldots,m\}$ be the indices of covered intervals. The passenger's utility is then
\begin{align*}
U_i^{Partial}(x,y)
&=u_i(x,y)-\alpha\sum_{j\notin C} q(t_{j-1},t_j,\lambda)-\sum_{j\in C} p^*(t_{j-1},t_j) \\
&=u_i(x,y)-\alpha\sum_{j\notin C} q(t_{j-1},t_j,\lambda)-\alpha\sum_{j\in C} q(t_{j-1},t_j,\lambda) \\
&=u_i(x,y)-\alpha\sum_{j=1}^m q(t_{j-1},t_j,\lambda) \\
&=u_i(x,y)-\alpha q(x,y,\lambda),
\end{align*}
where the last equality uses the additivity of $q$ on adjacent intervals.

Therefore every direct ticketing strategy gives the same payoff, namely $u_i(x,y)-\alpha q(x,y,\lambda)$. In particular, buying the full ticket for the direct path is optimal for every passenger, so $p^*$ is incentive-compatible. The same calculation shows that every passenger's payoff under that choice is exactly $u_i(x,y)-\alpha q(x,y,\lambda)$, so $p^*$ is full-surplus-extracting.

It remains to show maximality among incentive-compatible prices. Consider any pricing scheme $\tilde p$ under which passengers optimally choose full payment. Fix a passenger traveling from $x$ to $y$, and again reduce attention to the direct path by Lemma~\ref{lem:directRoute}. Let $C$ be the set of intervals on which she buys coverage, with respect to a partition as above. Her fare payment is $\sum_{j\in C}\tilde p(t_{j-1},t_j)$. If
\[
\sum_{j\in C}\tilde p(t_{j-1},t_j)>\alpha\sum_{j\in C} q(t_{j-1},t_j,\lambda),
\]
then deviating to the same path but buying no ticket on those covered intervals would strictly increase her payoff. Therefore every incentive-compatible payment made by a passenger traveling from $x$ to $y$ is bounded above by
\[
\alpha\sum_{j=1}^m q(t_{j-1},t_j,\lambda)=\alpha q(x,y,\lambda)=p^*(x,y).
\]
Thus no incentive-compatible pricing scheme can generate more fare revenue from that passenger than $p^*$ does. Integrating over all origin-destination pairs, no incentive-compatible pricing scheme yields higher total fare revenue than $p^*$. Hence $p^*$ is maximal among incentive-compatible pricing schemes.
\end{proof}

The assumption of risk neutrality is convenient for identifying the highest incentive-compatible price in expected-value terms. For a risk-averse passenger, the certain payment $p^*(x,y)=\alpha q(x,y,\lambda)$ is weakly preferred to any lottery with the same mean generated by traveling uncovered, so the pricing scheme $p^*$ remains incentive-compatible under risk aversion. Indeed, the maximal incentive-compatible price is strictly \emph{higher} under risk aversion: because a risk-averse rider values the uncertain fine lottery below its expected value, a fare above $\alpha q(x,y,\lambda)$ can be charged while still deterring evasion. The risk-neutral prices derived here are therefore conservative lower bounds on what could be sustained with risk-averse riders, and the revenue comparisons in the application understate the revenue achievable under a more realistic behavioral assumption.

Next, we consider the distribution of inspectors $\lambda$. The focus of this paper is on prices and the incentives they induce on strategic passengers, and therefore we have taken the inspection distribution as given. However, understanding the revenue-maximizing inspection distribution is useful for interpreting whether the scenarios we consider are reasonable.

A distribution of inspectors $\lambda$ is \textbf{revenue-maximizing} if it maximizes the transportation authority's revenues under a revenue-maximizing pricing scheme. The next proposition shows that the problem is particularly simple once Proposition~\ref{prop1} is established.

\begin{prop}
\label{prop:optimalDistributionInspectors}
A feasible distribution of inspectors $\lambda$ is revenue-maximizing only if, for any two points $x,y\in(0,1)$,
\[
\phi'\left(\lambda\left(x\right)\right)=\phi'\left(\lambda\left(y\right)\right).
\]
Conversely, any feasible distribution that satisfies this condition is revenue-maximizing.
\end{prop}

\begin{proof}
From Proposition~\ref{prop1}, the revenue-maximizing pricing scheme is $p^*(x,y)=\alpha q(x,y,\lambda)$. Therefore,
\begin{align*}
\pi
&=\alpha \int_0^1 \int_x^1 d(x,y)\left[\int_x^y \frac{\phi\left(\lambda(a)\right)}{d^{pass}(a)}da\right]dy\,dx
+\alpha \int_0^1 \int_0^x d(x,y)\left[\int_y^x \frac{\phi\left(\lambda(a)\right)}{d^{pass}(a)}da\right]dy\,dx.
\end{align*}
Applying Fubini's theorem and using the definition of $d^{pass}(a)$ yields
\begin{align*}
\pi
&=\alpha \int_0^1 \frac{\phi\left(\lambda(a)\right)}{d^{pass}(a)}
\left[\int_0^a \int_a^1 d(x,y)dy\,dx + \int_a^1 \int_0^a d(x,y)dy\,dx\right] da \\
&=\alpha \int_0^1 \phi\left(\lambda(a)\right)da.
\end{align*}
Hence the transportation authority solves
\[
\max_{\lambda\ge 0}\; \alpha \int_0^1 \phi\left(\lambda(a)\right)da
\quad\text{subject to}\quad
\int_0^1 \lambda(a)da=1.
\]
The feasible set is convex, and the objective is concave because $\phi$ is concave. Therefore the Lagrangian first-order conditions are necessary and sufficient. The Lagrangian is
\[
\mathcal{L}=\alpha\int_0^1 \phi\left(\lambda(a)\right)da + \Lambda\left[1-\int_0^1 \lambda(a)da\right].
\]
Its first-order condition at each $a\in(0,1)$ is
\[
\alpha\phi'\left(\lambda(a)\right)-\Lambda=0,
\]
which implies
\[
\phi'\left(\lambda(x)\right)=\phi'\left(\lambda(y)\right)=\Lambda/\alpha
\]
for every $x,y\in(0,1)$. This proves necessity. Since the optimization problem is concave, any feasible distribution satisfying these first-order conditions is also a global maximizer, which proves sufficiency.
\end{proof}

Proposition~\ref{prop:optimalDistributionInspectors} says that the revenue-maximizing inspection distribution equalizes the marginal effectiveness of inspectors at every point. Two immediate corollaries follow.

\begin{cor}
\label{cor:OptimalInspectorsStrictlyConcave}
If $\phi$ is strictly concave, the unique revenue-maximizing distribution of inspectors is uniform: $\lambda^*(a)=1$ for every $a\in[0,1]$.
\end{cor}

\begin{proof}
Under strict concavity, $\phi'$ is injective. Hence Proposition~\ref{prop:optimalDistributionInspectors} implies that any revenue-maximizing distribution must satisfy $\lambda^*(x)=\lambda^*(y)$ for all $x,y$. Since $\int_0^1 \lambda^*(a)da=1$, the only such distribution is $\lambda^*(a)=1$ for every $a\in[0,1]$.
\end{proof}

\begin{cor}
\label{cor:OptimalInspectorsLinear}
If $\phi(\lambda)=k\lambda$ for some constant $k>0$, then any feasible distribution of inspectors $\lambda^*$ with $\int_0^1 \lambda(x)dx=1$ is revenue-equivalent.
\end{cor}

\begin{proof}
If $\phi(\lambda)=k\lambda$, then
\[
\pi=\alpha\int_0^1 k\lambda(a)da=\alpha k,
\]
which is independent of the distribution $\lambda$.
\end{proof}

Corollary~\ref{cor:OptimalInspectorsLinear} implies that if inspection exposure is proportional to the inspector-passenger ratio, revenue maximization imposes no further restriction on how inspectors are distributed along the line.

\section{Networks\label{section:Networks}} 

We now adapt the same logic to a finite connected undirected network $G=(\mathcal{N},E)$. The notation changes accordingly: $x$ and $y$ now denote discrete nodes in $\mathcal{N}$ rather than points on the open interval $(0,1)$; the continuous passenger density $d^{pass}(a)$ is replaced by discrete edge flows $\sigma^{\mathrm{pass}}_e$; and the integral inspection-exposure function $q(x,y,\lambda)$ is replaced by a per-edge ratio $Q_e(\lambda)$. The economic logic, however, carries over directly. The environment again contains a positive mass of nonatomic passengers, so each individual rider takes aggregate traffic conditions as given. For each origin $x\in \mathcal{N}$ and destination $y\in \mathcal{N}$, there is exogenous demand $D_{(x,y)}\ge 0$. As in the line model, demand is treated as fixed so that the object of interest is truthful payment conditional on travel.

What changes relative to the line model is that a trip need not be associated with a unique route. A passenger traveling from $x$ to $y$ may have several feasible paths available. A path is a finite ordered list of nodes
\[
\sigma=\left(z^0,z^1,\ldots,z^m\right)
\]
with $z^0=x$, $z^m=y$, and $(z^{h-1},z^h)\in E$ for every $h=1,\ldots,m$. We write $\mathcal{S}(\sigma)$ for the ordered multiset of traveled edges. By the same reasoning as in the line model, cycles are weakly dominated when prices and fines are nonnegative, so we may restrict attention to routes with no repeated nodes. We call such a path a \textbf{feasible route} for the origin-destination pair $(x,y)$.

The transportation authority does not physically control access. Instead, it offers tickets for trips between origins and destinations. In the network model, the primitive price object is therefore not an edge price but a \textbf{route ticket}: for every feasible route $\sigma$, let $P(\sigma)\in\mathbb{R}_+$ denote the fare for buying full coverage of that route. A passenger may either buy the full ticket for the route she uses, buy a collection of tickets for subpaths, or leave some edges uncovered.

To define inspection exposure, let $\sigma^{\mathrm{pass}}_{e}$ denote the mass of passengers who traverse edge $e\in E$ in either direction. To avoid undefined expressions on edges with zero strategic flow, we introduce a tremble parameter $\varepsilon>0$, interpreted as an arbitrarily small background flow on every edge. If inspectors are distributed across edges according to $\lambda=(\lambda_e)_{e\in E}$, with $\sum_{e\in E}\lambda_e=1$, then the expected number of inspections on edge $e$ is
\begin{equation}
  Q_{e}(\lambda)=\frac{\phi\left(\lambda_e\right)}{\sigma^{\mathrm{pass}}_{e}+\varepsilon}.
\end{equation}
As before, $\phi$ is increasing, differentiable, and concave, and the formula reflects the same inspection technology: each inspection event on an edge selects one passenger uniformly at random. The parameter $\varepsilon$ is purely a regularization device ensuring that exposure is well defined even when an edge receives zero strategic flow.

Given a path $\sigma$, the expected fine from traveling uncovered on all of its edges is
\[
\alpha\sum_{e\in \mathcal{S}(\sigma)} Q_{e}(\lambda).
\]
More generally, if a passenger covers some subpaths and leaves others uncovered, her expected utility is travel utility minus ticket payments minus expected fines on the uncovered edges. Because passengers are nonatomic, no individual rider affects aggregate edge flows, so each rider treats the collection $(Q_e)_{e\in E}$ as given.

The appropriate definition of incentive compatibility is therefore parallel to the line model. A route-ticket pricing scheme is incentive-compatible if, given aggregate edge exposures, almost every passenger finds it optimal to buy the full ticket for the route she uses rather than substitute uncovered travel or a cheaper collection of subpath tickets.

\begin{lem}
\label{lem:simplePathNetwork}
Fix any route-ticket pricing scheme and any inspection exposure function on edges. For a passenger traveling from $x$ to $y$ in a network, every path that contains a cycle is weakly dominated by a route from $x$ to $y$ with no repeated nodes. The domination is strict whenever the original path traverses some edge more than once.
\end{lem}

\begin{proof}
Any path that contains a cycle can be shortened by deleting the cycle while keeping the same origin and destination. Removing a cycle weakly decreases ticket expenditure and weakly decreases expected fines, because both are nonnegative on every traveled edge. Travel utility remains $u_i(x,y)$ because origin and destination do not change. Repeating this argument until no cycle remains yields a route with no repeated nodes that weakly dominates the original path. The domination is strict whenever some edge is removed from the trip.
\end{proof}

\begin{prop}
\label{prop:existenceofpricenetwork}
For any finite connected network $G=(\mathcal{N},E)$, any monitor allocation $\lambda$ with $\lambda_e>0$ for every edge $e\in E$, and any tremble $\varepsilon>0$, there exists an incentive-compatible route-ticket pricing scheme. Moreover, this pricing scheme admits an additive representation over edges: for every feasible route $\sigma$,
\begin{equation}
\label{eq:ICPriceInNerwork}
P^{*,\varepsilon}(\sigma)=\alpha\sum_{e\in \mathcal{S}(\sigma)} Q_e(\lambda).
\end{equation}
\end{prop}

\begin{proof}
The proof proceeds in two steps. We construct an auxiliary game in which passengers' only decision is which route to take. Prices are not fixed exogenously but are determined endogenously by the aggregate route flows via equation~\eqref{eq:ICPriceInNerwork}: each route's fare equals the sum of its edge-level inspection exposures, which depend on how many passengers use each edge. A passenger who selects a route simply pays that route's price; no alternative ticket strategy is available in the auxiliary game. Part~(b) shows, using \cite{schmeidler1973equilibrium}, that this game has a Nash equilibrium. That equilibrium determines the pricing schedule~\eqref{eq:ICPriceInNerwork}, and the equilibrium condition guarantees that no passenger wants to switch routes given the induced prices. Part~(a) then verifies that, given those prices, no passenger can reduce her cost by altering her ticket purchases on her chosen route: any combination of partial coverage and uncovered travel incurs the same expected total cost as the full ticket.  Together, the two steps deliver both an incentive-compatible pricing schedule and the equilibrium route pattern that supports it.

\medskip\noindent\textit{Part~(a): optimality of the full ticket conditional on a path.}
Fix any aggregate edge flows and define the route-ticket price of path $\sigma$ by equation~\eqref{eq:ICPriceInNerwork}. Take any passenger following path $\sigma$. If she buys the full ticket $P^{*,\varepsilon}(\sigma)$, her payment equals $\alpha\sum_{e\in\mathcal{S}(\sigma)} Q_e(\lambda)$. If instead she covers a subset of edges and leaves the rest uncovered, her total expected cost---ticket payments plus expected fines---is exactly the same, because both components add edge by edge. Hence, conditional on a path, buying the full ticket is optimal.

\medskip\noindent\textit{Part~(b): existence of equilibrium route flows.}
We formulate route choice as a nonatomic game. For each origin-destination pair $(x,y)$ with $D_{(x,y)}>0$, there is a continuum of passengers of mass $D_{(x,y)}$. The player space $(\mathcal{I},\mathcal{A},\mu)$ is the disjoint union of these continua, where $\mathcal{A}$ is the induced $\sigma$-algebra on $\mathcal{I}$ and $\mu$ is the measure that assigns mass $D_{(x,y)}$ to the $(x,y)$-group. Because each group is nonatomic and there are finitely many groups, $\mu$ is a finite nonatomic measure.

Let $\Sigma$ denote the set of all feasible routes in~$G$. Because $G$ is finite, $\Sigma$ is finite. Every passenger chooses an element of $\Sigma$.

A strategy profile is a measurable function $s:\mathcal{I}\to\Sigma$. The induced aggregate route frequencies $f_\sigma=\mu(\{i:s(i)=\sigma\})$ determine edge flows $\sigma^{\mathrm{pass}}_e=\sum_{\sigma:\,e\in\mathcal{S}(\sigma)} f_\sigma$. The payoff to a passenger with origin-destination $(x,y)$ from choosing a feasible route $\sigma$ connecting $x$ to $y$ is
\[
u_i(x,y)-\alpha\sum_{e\in\mathcal{S}(\sigma)} Q_e(\lambda)
=u_i(x,y)-\alpha\sum_{e\in\mathcal{S}(\sigma)} \frac{\phi(\lambda_e)}{\sigma^{\mathrm{pass}}_e+\varepsilon}.
\]
For a passenger with origin-destination $(x,y)$, any route in $\Sigma$ that does not connect $x$ to $y$ is assigned a payoff $-M_{xy}$, where $M_{xy}>0$ is chosen large enough that every such route yields strictly lower payoff than any feasible route connecting $x$ to $y$. This is possible because $\Sigma$ is finite and, with $\varepsilon>0$, every payoff from a feasible route is finite.
Because $\Sigma$ is finite, the distribution of strategies is characterized by the vector $(f_\sigma)_{\sigma\in\Sigma}\in\mathbb{R}^{|\Sigma|}$. Edge flows are affine functions of this vector, and since $\varepsilon>0$, each $Q_e(\lambda)=\phi(\lambda_e)/(\sigma^{\mathrm{pass}}_e+\varepsilon)$ is continuous in edge flows. Every passenger's payoff is therefore a continuous function of the strategy distribution. Finally, $u_i(x,y)$ depends on $i$ only through the origin-destination pair, and the partition of~$\mathcal{I}$ into OD groups is measurable, so the payoff function is measurable with respect to $\mathcal{A}$ in~$i$ for every fixed strategy distribution.

In summary, the game has a finite action set, a nonatomic player space, payoffs that are continuous in the strategy distribution, and payoffs that are measurable in the player index. The main theorem of \cite{schmeidler1973equilibrium} therefore guarantees the existence of a pure-strategy Nash equilibrium: a measurable $s^*:\mathcal{I}\to\Sigma$ such that $\mu$-almost every passenger chooses a payoff-maximizing path given the aggregate flows induced by~$s^*$.\footnote{That is, except possibly for a set of passengers with $\mu$-measure zero.} Part~(a) then implies that, at those equilibrium flows, equation~\eqref{eq:ICPriceInNerwork} defines a pricing scheme under which every passenger who uses a route finds the full ticket for that route optimal. This establishes the existence of an incentive-compatible route-ticket pricing scheme with the claimed additive edge representation.

\end{proof}

Proposition~\ref{prop:existenceofpricenetwork} is the network analog of Proposition~\ref{prop1}. The tremble $\varepsilon>0$ is a technical device used to establish existence: it ensures that payoffs are continuous in the strategy distribution, which is needed for the fixed-point argument. Once equilibrium edge flows are known to be strictly positive, the regularization plays no further role and can be removed.\footnote{The per-edge price $Q_e=\phi(\lambda_e)/(\sigma_e^{\mathrm{pass}}+\varepsilon)$ is \emph{decreasing} in edge flow: more passengers on an edge lower the price each pays. This is the opposite of a congestion externality. The resulting positive externality means concentrated flow is self-reinforcing, so the pure-strategy equilibria guaranteed by Schmeidler's theorem are \emph{corner} equilibria in which each OD pair sends all demand along a single route. Starting from an initial assignment (e.g.\ shortest paths), the implied prices can be computed and OD pairs reassigned to their cheapest route; this process converges, computationally, in very few iterations.}

\begin{cor}\label{cor:epsilonzero}
For any equilibrium of the $\varepsilon$-perturbed game, let $\Sigma^*\subseteq\Sigma$ denote the set of routes that carry strictly positive flow. As $\varepsilon\to 0$, the incentive-compatible prices $P^{*,\varepsilon}(\sigma)$ converge to $P^{*,0}(\sigma)=\alpha\sum_{e\in\mathcal{S}(\sigma)}\phi(\lambda_e)/\sigma_e^{\mathrm{pass}}$ for every $\sigma\in\Sigma^*$, and $P^{*,0}$ is incentive-compatible at the limiting flows. For routes not in $\Sigma^*$, at least one edge carries zero limiting flow, so $P^{*,\varepsilon}(\sigma)\to+\infty$; since incentive-compatible prices are upper bounds on fares, a diverging upper bound means that any finite price for such a route satisfies incentive compatibility.
\end{cor}

\begin{proof}
Let $(\sigma_e^{\mathrm{pass},\varepsilon})_{e\in E}$ be equilibrium edge flows for a given $\varepsilon>0$. These flows lie in a compact set (each is bounded by total demand), so some subsequence converges to a limit $(\sigma_e^{\mathrm{pass},0})_{e\in E}$. For any route $\sigma\in\Sigma^*$, every edge in $\mathcal{S}(\sigma)$ carries strictly positive limiting flow, so $\phi(\lambda_e)/(\sigma_e^{\mathrm{pass},\varepsilon}+\varepsilon)\to\phi(\lambda_e)/\sigma_e^{\mathrm{pass},0}$ and the route-ticket prices converge. Part~(a) of the proof of Proposition~\ref{prop:existenceofpricenetwork} uses only the additive structure of prices and fines and does not depend on~$\varepsilon$, so $P^{*,0}$ is incentive-compatible at the limiting flows. For any route $\sigma\notin\Sigma^*$, some edge $e\in\mathcal{S}(\sigma)$ has $\sigma_e^{\mathrm{pass},0}=0$, so $P^{*,\varepsilon}(\sigma)\to+\infty$. Because incentive-compatible prices are upper bounds on what can be charged, a diverging upper bound means that any finite fare for such a route satisfies the IC condition.
\end{proof}

In the application below, the canonical-route assignment produces strictly positive flow on every edge in the network, so the unperturbed formula $P^{*,0}$ applies directly.

\section{Application: the Washington DC Metro\label{section::application}}

The Washington DC Metro is a useful test case for several reasons. It is a large, multi-line network with shared segments where different routes overlap, which is exactly the structure that makes incentive-compatible pricing non-trivial. Origin-destination ridership counts and fare tables are available, making it possible to calibrate the model with real traffic and pricing data.

In this section, we use these data as an example of the impact that our incentive-compatible pricing scheme could have in a real-life public transport system. For each origin-destination pair, we allow all feasible routes in the network. The objective is not to estimate a full demand-and-routing system, but to evaluate how large the revenue consequences of incentive-aware pricing can be under proof-of-payment when prices and strategic deviations are computed at route-equilibrium flows.

Our objectives are as follows. First, we calibrate the fine $\alpha$ and the resulting incentive-compatible prices by requiring revenue neutrality: we choose $\alpha$ so that the total revenue the DC Metro would collect under proof-of-payment with incentive-compatible prices exactly equals the revenue currently obtained under the gated pricing schedule. This requirement pins down $\alpha$ and, through equation~\eqref{eq:ICPriceInNerwork}, the entire IC price schedule. The revenue-neutral $\alpha$ then serves as a natural benchmark: it is the fine level at which switching from barrier-based enforcement to proof-of-payment, while adopting IC prices, leaves total revenue unchanged. We check whether the implied fine is economically reasonable given published fine levels.

Second, we perform counterfactual exercises which evaluate the impact on fare evasion and revenue that would result from switching from the current barrier-based system to a proof-of-payment system with random inspection while keeping the 2019 prices.

Third, we show how adjusting the current pricing scheme via the use of incentive-compatible prices could substantially increase revenue, by reducing some of the prices at some of the segments of the network.

\subsection{Data}

We use two datasets in our analysis. One is a station-to-station passengers count for the month of May 2012.\footnote{Source: Washington Metropolitan Area Transit Authority.} It contains, for every pair of metro stations, the average daily number of passengers traveling from one to the other. These are separated into four parts of the day: AM peak (opening to 9:30am), Midday (9:30am to 3:00pm), PM peak (3:00pm to 7:00pm) and Evening (7:00pm to midnight). We manually replicated the network structure of the subway system.

The second is the origin-destination pricing, as of April 2019. The DC Metro uses a pricing scheme in which the cost of the ticket depends on the station of departure and of destination, and also on whether it takes place at Peak or Off-Peak times. Prices varied from US\$2.00 to US\$6.00, and there were 76 different prices depending on the origin, destination, and time.

In our analysis, we use the 2012 traffic data together with the 2019 pricing information.\footnote{The reason for using this combination was data availability: we only had the traffic data from May 2012, but no past pricing data.} Since the objective of our simulation exercise is to evaluate the model with realistic values, this combination fits our purposes, despite not delivering exact prices for 2012 or 2019. In addition to these, we need values for two more parameters in our model: $\alpha$ (the value of the fine charged to passengers caught with an incorrect ticket) and $\lambda$ (the total mass of inspectors).

For the total mass of inspectors, we consider the case of Berlin public transport. Berlin also has a large-scale public transport system, but differently from the DC Metro, it uses a proof-of-payment system, as in our model. In 2018, Berlin had 120 inspectors and 2.9 million passengers per day.\footnote{Sources: \url{http://www.exberliner.com/features/zeitgeist/controllers-out-of-control} and \url{http://www.bvg.de/en/Company/Profile/Structure--facts}} The DC Metro, in our data, has 724,156 passengers per day. Therefore, setting the number of inspectors in DC to 30 would make the proportion of inspectors per passenger approximately the same.

Based on published travel times, it takes about 70 minutes to ride 26 stations in a line in the DC Metro, or roughly 2.69 minutes per station on average. If we consider the four periods in the data as equally distributed, each period covers 6 hours (360 minutes), so each inspector can potentially inspect 134 stations per period. The DC Metro network has 88 edges between stations, meaning each inspector could in principle pass 1.52 times through each edge during a period under uniform distribution. If we have 30 inspectors, then each edge is inspected by 45.60 inspectors over the six-hour period. The total mass of inspectors to be distributed is, therefore, $\lambda^{Total}=\sum_{(x,y) \in E}\lambda_{(x,y)}=4,013$.

By using the number of passengers traveling from each origin and destination and the price for those trips, we can calculate what would be the revenue that the DC Metro obtained, on average, in May, during each period of the day. Table~\ref{table:TrafficAndRevenue2019Prices} shows the results of these calculations. After excluding origin-destination pairs without a matching fare record---mainly same-station rows---the cleaned input layer covers 86 stations, 88 network edges, and 7,310 priced OD pairs. We see that 724,156 subway trips are made on average per day, resulting in a daily revenue of $\pi^{Total}=\text{US\$}2{,}353{,}948$.

\begin{table}[t]
    \centering
    \begin{tabular}{l|c|c|c|c|c}
     &  \textbf{AM Peak} & \textbf{Midday} & \textbf{PM Peak} & \textbf{Evening} & \textbf{Total}\\
     \hline
     \hline
     
   Total Traffic  & 235,150 & 141,498 & 257,550 & 89,958 & 724,156\\
   Revenue & \$845,138 & \$373,262 & \$890,785 & \$244,764 & \$2,353,948 \\
    \hline
   \hline 
\end{tabular}
    \caption{Passenger traffic and revenue under the 2019 pricing scheme}
    \label{table:TrafficAndRevenue2019Prices}
\end{table}

The origin-destination data gives us the number of passengers who travel from each origin to each destination. In order to produce the incentive-compatible prices, however, we need edge flows rather than only OD totals. We therefore enumerate, for each origin-destination pair, all feasible routes and assign each OD pair to the cheapest route-and-ticket strategy. As noted in the Remark following Proposition~\ref{prop:existenceofpricenetwork}, the incentive-compatible edge price $Q_e$ is decreasing in edge flow, so concentrated flow is self-reinforcing and the equilibrium is a corner solution in which each OD pair uses a single route. To compute it, we assign every OD pair to the shortest route, aggregate edge flows, compute the implied prices, and reassign any OD pair that would prefer to switch. Because concentrated flow is self-reinforcing, this iterative verification converges within a few iterations. The two monitoring configurations used to compute edge prices are described in Section~\ref{subsec:ICPrices} below. Whenever two routes are tied in distance or total route-and-ticket cost, the tie is broken by a fixed exogenous priority ordering over edges: among the tied routes, we select the one whose highest-priority edge ranks first under that ordering.

\subsection{Fine for fare evasion}

We now derive the revenue-neutral value of $\alpha$. We set $\phi(\lambda)=\lambda$, so the expected number of inspections on an edge equals the inspector-to-passenger ratio on that edge. Under this assumption, the IC price for edge $(x,y)$ is

\begin{equation}
P_{(x,y)}=\frac{\lambda_{(x,y)}}{\sigma_{(x,y)}^{pass}}\alpha.
\end{equation}

The revenue collected from passengers traversing edge $(x,y)$ is therefore $\pi_{(x,y)}=P_{(x,y)}\sigma_{(x,y)}^{pass}=\alpha\lambda_{(x,y)}$: the edge flows cancel and per-edge revenue depends only on the inspector allocation. Notice that \textit{fines are not considered revenue}, but enter only as parameters of the IC price formula. Summing over all edges,

\begin{equation}
\pi^{Total}=\sum_{(x,y)\in E}\pi_{(x,y)}=\alpha\sum_{(x,y)\in E}\lambda_{(x,y)}=\alpha\lambda^{Total},
\end{equation}
and therefore $\alpha=\pi^{Total}/\lambda^{Total}$. The revenue-neutral fine depends only on total revenue and total inspector mass.

Applying this formula period by period, using the per-period revenues from Table~\ref{table:TrafficAndRevenue2019Prices} and $\lambda^{Total}=4{,}013$, we obtain $\alpha=\text{US\$}210.60$, US\$93.01, US\$221.97, and US\$60.99 for the AM Peak, Midday, PM Peak, and Evening periods, respectively. Interestingly, the size of the fine required is not too much different from the one actually charged under the DC system. Before decriminalization of fare evasion in May 2019, police could issue criminal citations up to US\$300 and even jail people for 10 days. Under the new law, it becomes a civil penalty with a maximum fine of US\$50.\footnote{D.C. Law 22-310. Fare Evasion Decriminalization Amendment Act of 2018: \url{https://code.dccouncil.us/us/dc/council/laws/22-310}}

\subsection{Incentive-compatible prices}
\label{subsec:ICPrices}

Given these values for $\alpha$, we can then go back to Equation~\eqref{eq:ICPriceInNerwork} and obtain the additive edge components of the incentive-compatible prices. We consider two configurations regarding the distribution of inspectors in the network. 

The first is the \textbf{uniform distribution} case, in which inspectors are uniformly distributed across the system. That is, for every edge $(x,y) \in E$, $\lambda_{(x,y)}=\bar{\lambda}$. Since we set $\phi(\lambda)=\lambda$, Corollary~\ref{cor:OptimalInspectorsLinear} implies that every inspector distribution yields the same revenue; uniform is therefore one of many revenue-equivalent configurations. The ticket price for an edge $(x,y)$ is:

\begin{equation}
P_{(x,y)}=\frac{\bar{\lambda}}{\lambda^{Total}} \frac{\pi^{Total}}{\sigma_{(x,y)}^{pass}}
\end{equation}

The second configuration that we considered was the \textbf{proportional distribution}, where the mass of inspectors in an edge $(x,y)$ of the network is proportional to the total flow of passengers passing by it. That is:

\begin{equation}
\lambda_{(x,y)}=\frac{\lambda^{Total}\sigma^{pass}_{(x,y)}}{\sigma^{Total}}\text{, where }\sigma^{Total}=\sum_{(x,y)\in E}\sigma^{pass}_{(x,y)}
\end{equation}

But then, the ticket prices for an edge $(x,y)$ under $\alpha=\pi^{Total}/\lambda^{Total}$ become:

\begin{equation}
P_{(x,y)}=\frac{\lambda^{Total}\frac{\sigma_{(x,y)}^{pass}}{\sigma^{Total}}}{\lambda^{Total}} \frac{\pi^{Total}}{\sigma_{(x,y)}^{pass}}=\frac{\sigma_{(x,y)}^{pass}}{\sigma^{Total}} \frac{\pi^{Total}}{\sigma_{(x,y)}^{pass}}=\frac{\pi^{Total}}{\sigma^{Total}}
\end{equation}

That is, with proportional monitoring, the incentive-compatible price is the same for each edge. 

In the application, each origin-destination pair is assigned its feasible routes and prices are computed from the resulting route-equilibrium flows. Under uniform monitoring, the route-ticket price for an OD pair is the minimum additive cost among its feasible routes under the equilibrium edge penalties. Under proportional monitoring every edge carries the same price within a period, so the route-ticket price depends only on route length. This yields a tighter IC-price distribution under proportional monitoring, whereas uniform monitoring produces more dispersion because inspection exposure varies with the equilibrium distribution of passengers across edges.

\begin{table}[t]
    \centering
    \begin{tabular}{l|c|c|c|c}
     &  \textbf{AM Peak} & \textbf{Midday} & \textbf{PM Peak} & \textbf{Evening} \\
     \hline
     \hline
   \multicolumn{5}{c}{\textbf{DC Metro 2019 pricing}}     \\
   \hline
   Minimum ticket price & \$2.25 & \$2.00 & \$2.25 & \$2.00 \\
   Median ticket price & \$4.20 & \$3.40 & \$4.20 & \$3.40 \\
   Maximum ticket price & \$6.00 & \$3.85 & \$6.00 & \$3.85 \\
   \hline
   \multicolumn{5}{c}{\textbf{Incentive-compatible pricing - Uniform monitoring}}\\
   \hline
  Minimum ticket price & \$0.18 & \$0.12 & \$0.17 & \$0.14 \\
  Median ticket price & \$4.30 & \$3.82 & \$4.46 & \$3.66 \\
  Maximum ticket price & \$17.45 & \$16.79 & \$18.25 & \$15.06 \\
   \hline
   
  Trips where $IC<DC$ & 68.48\% & 67.68\% & 69.12\% & 68.47\% \\ 
   \hline
   \hline
   \multicolumn{5}{c}{\textbf{Incentive-compatible pricing - Proportional monitoring}}\\
   \hline
   Minimum ticket price & \$0.44 & \$0.38 & \$0.45 & \$0.36 \\
   Median ticket price & \$4.82 & \$4.15 & \$4.93 & \$3.99 \\
   Maximum ticket price & \$11.83 & \$10.18 & \$12.09 & \$9.79 \\
   \hline
   
   Trips where $IC<DC$ & 54.13\% & 55.82\% & 55.18\% & 55.47\% \\ 

   \hline
   \hline 
\end{tabular}
    \caption{DC Metro 2019 and incentive-compatible (IC) prices under route-equilibrium flows.}
    \label{table:IncentiveCompatiblePrices}
\end{table}

  Some summary statistics from both the DC Metro 2019 pricing scheme and the incentive-compatible prices are shown in Table~\ref{table:IncentiveCompatiblePrices}. The central tendency of fares is broadly similar across the two systems, but the range of prices differs substantially. While the highest difference between the cost of two tickets under the 2019 DC pricing scheme is US\$4.00, for the incentive-compatible prices that difference jumps to more than US\$18.00.

  Under uniform monitoring, the median IC price lies close to the median DC fare in every period, but a larger share of trips becomes cheaper than under the current fare system than under proportional monitoring. Under proportional monitoring, by contrast, the IC prices are more tightly concentrated because every traversed edge carries the same price within a period.

Here it is important to remember the meaning of incentive-compatible prices: they are the highest prices that guarantee that passengers will pay the full fare. Therefore, any price below those still guarantees that all passengers will pay. Perhaps the main reason why these values can be as high as more than US\$18.00 is that some prices have to be very low. The reason for this is intuitive: if a passenger wants to make a short and busy trip, say from one station to the next, the likelihood that she will face an inspector is relatively low. Therefore, if the price is not low as well, a ``rational'' passenger will prefer to take the risk. Indeed, if the traffic pattern is such that many passengers make short trips, a random checks system will require a high value for $\alpha$ in order to have incentive-compatible prices that yield good revenues.

While some incentive-compatible prices are very high, under them most passengers would pay less than under the current DC Metro pricing. As shown in Table~\ref{table:IncentiveCompatiblePrices}, the majority of trips would be cheaper under the incentive-compatible prices in every period, especially under uniform monitoring. These prices should be interpreted as upper bounds in the following sense: they are the highest fares consistent with riders weakly preferring full payment to any underpayment strategy available on their feasible routes. Any uniformly lower fare schedule would also preserve incentive compatibility, though at the cost of lower revenue.

Figure~\ref{fig:priceRatioMaps} provides a visual summary of the relationship between incentive-compatible and DC~2019 prices. Because DC prices are set per origin-destination pair rather than per edge, there is no single ``DC price of an edge.'' To produce an edge-level comparison we proceed as follows. For each equilibrium-assigned trip from $a$ to $b$, we compare the full-trip IC price $P_{(a,b)}$ to the DC price $P^{DC}_{(a,b)}$ and record whether the passenger pays less under IC pricing. Each edge in the trip's route then receives that passenger's indicator. Each edge is colored by the passenger-weighted share of all equilibrium trips traversing it for which IC pricing is cheaper, aggregated over the four time periods within the corresponding monitoring scenario. A blue edge therefore indicates that the majority of passengers using that edge pay less under IC pricing; a red edge indicates that the majority pay more. Edge thickness reflects the scenario-specific equilibrium ridership on that edge. The two panels look strikingly different. Panel~(a) is predominantly blue: under uniform monitoring, inspectors are spread evenly, so heavily used central edges have very low per-passenger inspection exposure and therefore very low IC prices; most passengers who travel through those edges pay less than under the DC schedule, giving the map an overall blue cast. Panel~(b) is predominantly red: under proportional monitoring every edge carries the same per-edge IC price within a period, so a trip's total IC fare grows in proportion to its number of edges. The DC fare schedule also increases with distance, but with a compressed range of \$2.00--\$6.00, it rises much less steeply than a strictly additive per-edge charge; as a result, most medium-to-long routes cost more under IC than under DC. Under both monitoring scenarios, the central core of the network tends to be bluer than the periphery: far-out edges carry little traffic, so the per-passenger inspection exposure is high and the IC price on those edges is elevated, pushing IC fares above their DC counterparts for trips that traverse them. Because each long route traverses many edges, such trips contribute their ``more expensive'' flag to every edge along their path, making the proportional-monitoring map appear predominantly red even though the majority of passengers pay less overall (Table~\ref{table:IncentiveCompatiblePrices}); the trip-level share counts each passenger once regardless of route length, while the edge-level share counts each $k$-edge trip $k$ times. Under uniform monitoring, western edges tend to be bluer than eastern ones, reflecting the heavier commuter traffic on the Virginia side of the network.

\begin{figure}[t]
  \centering
  \includegraphics[width=\textwidth]{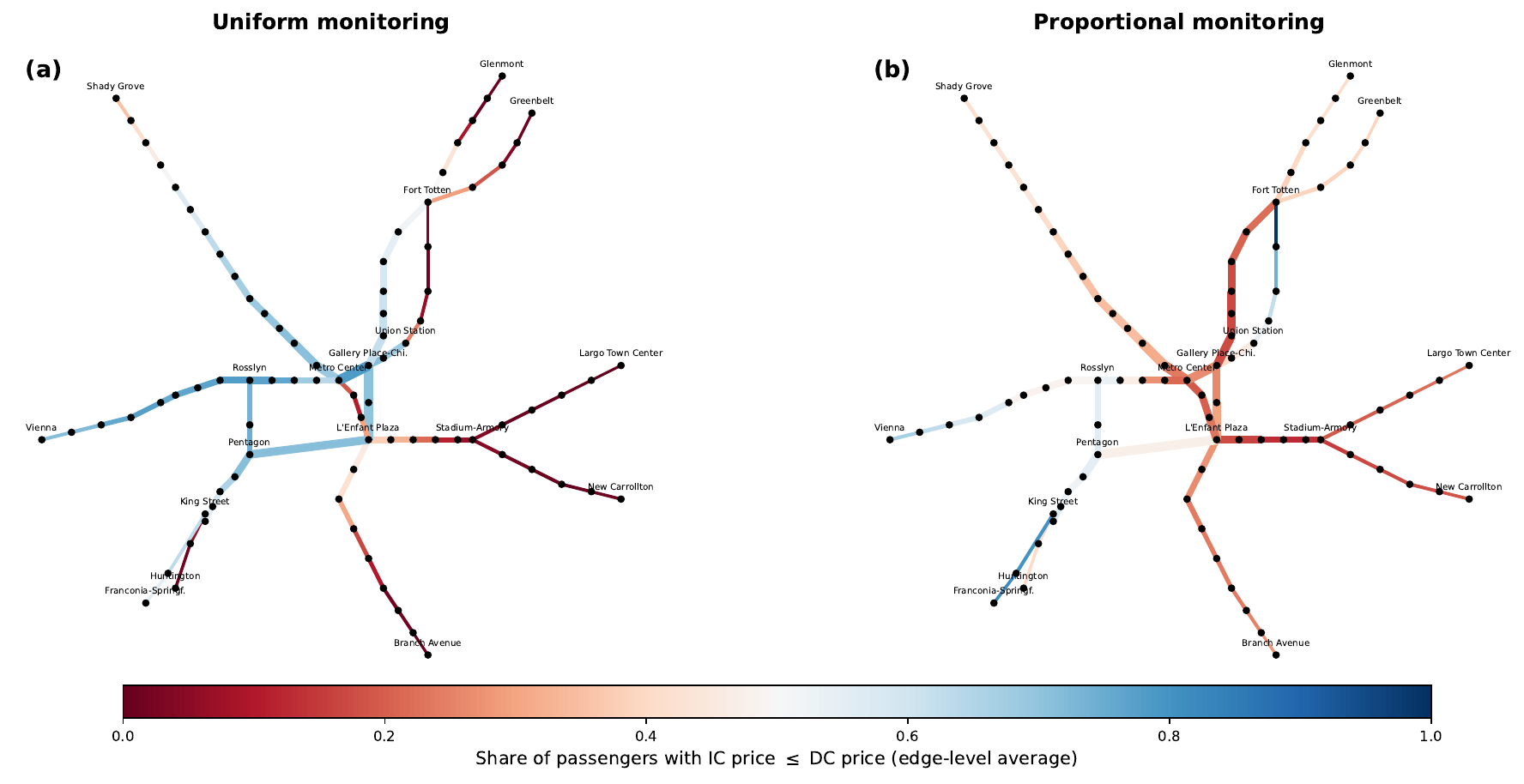}
  \caption{Schematic map of the DC Metro with edges colored by the passenger-weighted share of trips traversing that edge for which the IC price is at or below the DC~2019 price, under the route-equilibrium assignment for each monitoring scenario. Panel~(a) shows uniform monitoring and panel~(b) shows proportional monitoring. Blue edges (share above 0.5) indicate that the majority of passengers on that edge pay less under IC pricing; red edges (share below 0.5) indicate that the majority pay more. Both panels use the same color scale. Edge thickness is proportional to the equilibrium ridership on that edge in the corresponding panel.}
  \label{fig:priceRatioMaps}
\end{figure}

\subsection{DC prices under proof-of-payment system}
\label{subsec:DCPricesStrategicPassengers}

We next simulate a counterfactual in which the DC Metro switches to a proof-of-payment system with random inspection but keeps the 2019 prices. In other words, the authority changes the enforcement method but does not adapt the pricing to consider the incentives that they induce. The difference between these simulations and the 2019 revenue values can be seen as the impact that an evaluation of rider incentives has on the design of prices for public transport under random inspections.

\begin{table}[t]
  \centering
  \small
  \setlength{\tabcolsep}{4pt}
  \resizebox{\textwidth}{!}{%
  \begin{tabular}{l|c|c|c|c|c}
     &  \textbf{AM Peak} & \textbf{Midday} & \textbf{PM Peak} & \textbf{Evening} & \textbf{Total}\\
     \hline
     \hline
   \multicolumn{6}{c}{\textbf{Random inspections under DC 2019 Metro pricing - Uniform monitoring}}\\
   \hline 
  Partially Paid Trips & 4.23\% & 3.50\% & 4.27\% & 4.07\% & 4.08\% \\
  Trips without tickets & 68.19\% & 67.29\% & 68.45\% & 68.08\% & 68.09\% \\
  Fully Paid Trips  & 27.58\% & 29.21\% & 27.27\% & 27.85\% & 27.82\% \\
  Revenue  & \$314,348 & \$151,031 & \$342,020 & \$93,985 & \$901,384\\
  Losses due to partial tickets & \$6,528 & \$2,733 & \$7,130 & \$2,231 & \$18,622\\
  Losses due to no ticket purchased & \$524,262 & \$219,498 & \$541,634 & \$148,548 & \$1,433,942\\
   \hline
   
   \multicolumn{6}{c}{\textbf{Random inspections under DC 2019 Metro pricing - Proportional monitoring}}\\
   \hline
  Partially Paid Trips & 27.12\% & 14.47\% & 24.56\% & 13.87\% & 22.09\% \\
   Trips without tickets & 48.54\% & 53.14\% & 50.72\% & 54.18\% & 50.91\% \\
  Fully Paid Trips  & 24.34\% & 32.39\% & 24.72\% & 31.94\% & 26.99\% \\
  Revenue  & \$413,806 & \$188,348 & \$429,495 & \$118,670 & \$1,150,319\\
  Losses due to partial tickets & \$85,537 & \$15,781 & \$83,496 & \$9,625 & \$194,438\\
   Losses due to no ticket purchased & \$345,795 & \$169,133 & \$377,794 & \$116,469 & \$1,009,191\\
   \hline
   \hline 
\end{tabular}
  }
    \caption{Counterfactuals under DC pricing and random inspection with route-aware passenger deviations.}
    \label{table:Counterfactuals}
\end{table}

  Each passenger traveling from origin $a$ to destination~$b$ now compares all routes available for that OD pair and, on each route, solves the same cost-minimization problem via dynamic programming: she considers every way to partition that route into contiguous subpaths, choosing for each subpath whether to buy a ticket or ride uncovered, and selects the combination that minimizes her total expected cost (ticket prices plus expected fines for uncovered edges). The passenger then chooses the cheapest route-and-ticket strategy across routes: buying no ticket, buying one ticket for the full trip, or buying any collection of partial tickets on any admissible route. The values of $\alpha$ that we use are the ones calibrated in Section~\ref{subsec:ICPrices}, which guarantee that incentive-compatible prices yield the same revenue as the 2019 gated fares. Recall that the per-edge IC price $P_{(x,y)}=\alpha Q_{(x,y)}$ is also the expected fine from riding that edge uncovered: if the DC fare for a trip exceeds the sum of its per-edge IC prices, the passenger is better off evading.

% ---- EXAMPLE 1: fully uncovered (can be removed without affecting Example 2) ----
%To illustrate, consider a passenger traveling from Pentagon to L'Enfant Plaza during the AM Peak. The direct route uses a single, heavily-trafficked edge shared by the Yellow and Green lines. The DC fare is US\$2.25. Under uniform monitoring, the expected fine from riding this edge uncovered is only US\$0.22---equal to the~IC price for that edge and reflecting the very high ridership, which keeps the inspector-to-passenger ratio low. A passenger who rides without a ticket saves US\$2.03 relative to paying the fare; their expected total cost is nearly an order of magnitude below the ticket price. Every one of the 691 passengers making this trip in the AM Peak rationally chooses to ride uncovered.

% ---- EXAMPLE 2: partial ticket (can be removed without affecting Example 1) ----
To illustrate, consider a passenger traveling from Huntington, at the southern end of the Yellow Line, to L'Enfant Plaza during the AM Peak (8 edges, DC fare US\$3.80). Riding the full route uncovered would cost US\$4.67 in expected fines---exceeding the DC fare---so full evasion is not attractive. The optimal strategy is instead to buy a ticket for the first seven edges only (Huntington to Pentagon, DC price US\$3.25) and then ride the last, busy edge (Pentagon to L'Enfant Plaza) uncovered, paying an expected fine of US\$0.22. The total cost is US\$3.47, saving US\$0.33 relative to the full fare. This illustrates the partial-evasion logic: whenever a short, high-traffic stretch has very low inspection exposure, the cheapest strategy combines a ticket for the expensive portion with free-riding on the cheap one.

Table~\ref{table:Counterfactuals} shows the aggregate impact of using the 2019 DC Metro pricing with random inspections, for both uniform and proportional monitoring. In both cases, the amount of fare evasion and its impact on revenue are remarkably high. Under uniform monitoring, more than 68\% of trips are made without buying any ticket. Combined with some loss due to partial tickets, this results in a loss of about US\$1.45 million per day, or roughly 62\% compared with revenues when all passengers pay their full fares. Under proportional monitoring, the loss is about 51\% of benchmark revenue, which is lower than under uniform monitoring but still remarkably high. The fact that the two monitoring scenarios produce different revenue losses may seem to conflict with Corollary~\ref{cor:OptimalInspectorsLinear}, which establishes revenue equivalence across inspector distributions. The key is that Corollary~\ref{cor:OptimalInspectorsLinear} applies to incentive-compatible prices, where additivity ensures that every rider pays regardless of how inspectors are allocated. DC prices lack that additivity: whether a given trip is profitable to evade depends on the relationship between its fare and the expected fine on each edge, so changing the inspector distribution changes which trips are evaded and by how much. Appendix Tables~\ref{table:CounterfactualsThresholdUniform} and~\ref{table:CounterfactualsThresholdProportional} summarize analogous exercises when passengers deviate only if the reduction in ticket payments is at least US\$0.50, US\$1.00, US\$2.00, US\$3.00, US\$4.00, or US\$5.00.

\subsection{Incentives-adjusted DC prices}
\label{subsec:ICAdjustedDCPrices}

The results in Section~\ref{subsec:DCPricesStrategicPassengers} indicate that the impact on revenue that would result from switching to a proof-of-payment scheme with random inspections, without taking incentives into consideration, is potentially very high. One alternative that is available would be to use instead the incentive-compatible prices derived in Section~\ref{subsec:ICPrices}. Under these, every passenger would have the incentive to choose a full-payment strategy on an equilibrium route, and revenues would be the same as the ones before the change in the payment scheme.

One potential issue to this, however, is that some of the incentive-compatible prices that were derived are too high: some trips would cost more than US\$18.00, whereas under 2019 DC prices the highest price was US\$6.00. An alternative approach would be to use incentive-compatible prices, but limit them so that they are never higher than DC prices. That is, setting incentive-compatible prices bounded from above by 2019 DC prices.

At first sight, one might be tempted to, for every origin-destination $(x,y)$, setting the price $P^*_{(x,y)}=\min\left\{P^{DC}_{(x,y)},P_{(x,y)}\right\}$, where $P^{DC}_{(x,y)}$ and $P_{(x,y)}$ are respectively, 2019 DC and incentive-compatible prices. This, however, might result in $P^*$ not being an incentive-compatible pricing scheme. To see this, consider a network with three nodes: $x$, $y$, and $z$, with edges $(x,y)$ and $(y,z)$. Suppose that under DC prices, $P^{DC}_{(x,y)}=\$0.50$, $P^{DC}_{(y,z)}=\$1.00$, and $P^{DC}_{(x,z)}=\$1.50$, and that under incentive-compatible prices, $P_{(x,y)}=P_{(y,z)}=\$0.70$, and therefore $P_{(x,z)}=\$1.40$. If we simply set the prices to be the lowest between the two, we would set $P^*_{(x,y)}=\$0.50$, $P^*_{(y,z)}=\$0.70$ and $P^*_{(x,z)}=\$1.40$. Under $P^*$, however, a passenger traveling from $x$ to $z$ is better off by buying two tickets---$x$ to $y$ and $y$ to $z$. While she will still buy tickets for the entire journey, that configuration is not incentive-compatible, since the ticket purchased does not correspond to the origin-destination pair.

To correctly adjust DC prices to be incentive-compatible, one must eliminate every cheaper route-and-ticket deviation created by the cap itself. In the example above, the value of $P^*_{(x,z)}$ has to be further reduced to $P^*_{(x,z)}=\$1.20=\$0.50+\$0.70$. More generally, fix an admissible route $\sigma=(v_0,\ldots,v_K)$. The code solves a shortest-path dynamic program on that ordered list of nodes. At each node $v_i$, the passenger has two kinds of actions: either leave the next edge $(v_i,v_{i+1})$ uncovered and incur its expected penalty, or buy a ticket from $v_i$ to some later node $v_j$ and pay the currently adjusted origin-destination fare $P^*_{(v_i,v_j)}$. The full-trip ticket from $v_0$ to $v_K$ is excluded from this deviation problem, since the object being tested is precisely whether the current full-trip price is vulnerable to any alternative combination of subtrip tickets and uncovered edges. The dynamic program therefore returns the cheapest deviation on that route.

Starting from the naïve capped prices $\min\{P^{DC}_{(a,b)},P_{(a,b)}\}$, we run this dynamic program on every route between the origin and destination and take the cheapest deviation across routes. If that deviation is cheaper than the current full-trip price, we lower the full-trip price to that deviation cost; otherwise we leave it unchanged. The process repeats until no profitable deviations remain. The resulting fare schedule is deviation-proof, but it is not in general additive: after capping prices origin-destination by origin-destination and then adjusting them downward to remove profitable deviations, there is no reason for the final price of a trip to equal the sum of the final prices of its adjacent subtrips along the same route.\footnote{For a concrete illustration, consider the AM Peak trip from Franconia-Springfield to Shady Grove under uniform monitoring. The DC price is US\$6.00 (the system-wide cap) and the IC price is US\$11.05, so the adjusted price is $P^*=\$6.00$. The subtrips via Farragut West have adjusted prices $P^*(\text{Franconia-Springfield},\text{Farragut West})=\$5.65$ and $P^*(\text{Farragut West},\text{Shady Grove})=\$6.00$, summing to US\$11.65. No passenger would split a US\$6.00 ticket into two tickets costing US\$11.65, so the schedule is deviation-proof; yet $\$6.00\neq\$11.65$, so it is not additive.}

Table~\ref{table:ICLimitedFromAbove} shows the impact on revenue that results from adjusting DC prices to be incentive-compatible using the procedure described above. 

\begin{table}[t]
    \centering
  \small
  \setlength{\tabcolsep}{4pt}
  \resizebox{\textwidth}{!}{%
  \begin{tabular}{l|c|c|c|c|c}
     &  \textbf{AM Peak} & \textbf{Midday} & \textbf{PM Peak} & \textbf{Evening} & \textbf{Total}\\
     \hline
     \hline
   \multicolumn{6}{c}{\textbf{Uniform monitoring}}     \\
   \hline
  Revenue  & \$685,331 & \$278,557 & \$710,739 & \$185,205 & \$1,859,832\\
  Revenue loss percentage & 18.91\% & 25.37\% & 20.21\% & 24.33\% & 20.99\% \\
       \hline
     \hline
   \multicolumn{6}{c}{\textbf{Proportional monitoring}}     \\
   \hline
  Revenue  & \$729,796 & \$307,259 & \$763,908 & \$204,700 & \$2,005,663\\
   Revenue loss percentage & 13.65\% & 17.68\% & 14.24\% & 16.37\% & 14.80\% \\
   \hline
   \hline 
\end{tabular}
  }
    \caption{Incentive-compatible prices limited from above by DC prices and made deviation-proof.}
    \label{table:ICLimitedFromAbove}
\end{table}

Notice that, while the adjustment results in a significant loss of revenue, that pales in comparison to the scenarios simulated in Section~\ref{subsec:DCPricesStrategicPassengers}, where DC prices were used without any incentive adjustments. The comparison of the values in Tables~\ref{table:Counterfactuals} and~\ref{table:ICLimitedFromAbove} shows the increase in revenue that would result from \emph{lowering} the prices for certain paths, and therefore eliminating the incentive to travel without paying when traveling through them.

\section{Conclusion}

Proof-of-payment systems simplify transit operations, but they also create a pricing problem that does not arise under gated entry. Once riders can underpay by traveling uncovered or by buying only part of the trip, the fare schedule itself becomes a mechanism for shaping compliance. This paper characterizes the highest fares that remain incentive-compatible under random inspection. The maximal incentive-compatible fare is additive and equals the expected fine exposure from uncovered travel.

Applied to the Washington DC Metro under fixed demand and route-choice equilibrium, the quantitative implications are large. Keeping current fares after switching to proof-of-payment could reduce revenue by more than 60\%, while adjusting fares to eliminate profitable underpayment reduces that loss to between 15\% and 21\%, even though many fares become lower than under the existing system.

A few features of the analysis are worth keeping in mind when reading these results. Both the theory and the application hold travel demand fixed: the object is fare design that deters underpayment conditional on a travel need, not a general-equilibrium model of ridership. The revenue comparisons are also conservative because we assume risk neutrality; with risk-averse riders, the maximal incentive-compatible prices are strictly higher, so the revenue recovery from incentive-compatible pricing would be even larger. The application also combines 2012 origin-destination demand with 2019 fare tables---DC Metro ridership declined substantially in the intervening years, with average daily ridership in 2019 at about 86.6\% of the 2012 level\footnote{Source: https://www.wmata.com/initiatives/ridership-portal/Metrorail-Ridership-Summary.cfm}---so the precise revenue figures should be read with that mismatch in mind, even though the qualitative pattern is driven by the fare structure and inspection technology rather than the demand level.

Finally, incentive-compatible prices vary sharply across the network: they are highest on peripheral edges, where low ridership means each passenger faces a higher inspection probability, and lowest on busy central edges. This could raise equity concerns if peripheral commuters tend to be lower-income. It is worth stressing, however, that these prices are upper bounds on what the authority can charge while still deterring underpayment; the authority is free to set any fare at or below the incentive-compatible level, and the distributional pattern identifies where the feasibility constraint is tightest rather than prescribing where fares must be set.

\section*{Declaration of generative AI and AI-assisted technologies in the writing process}\label{sec:AIDeclaration}
During the preparation of this work the author(s) used ChatGPT in order to explore ideas, review and make adjustments to the simulations code, make minor edits and rephrase some sentences. After using this tool/service, the author(s) reviewed and edited the content as needed and take(s) full responsibility for the content of the publication.

\bibliographystyle{ecta}
\bibliography{subway}

\newpage
\appendix
\section*{Appendix: Map of the DC Metro in 2012}
\begin{figure}[h]
    \centering
    \includegraphics[scale=0.35]{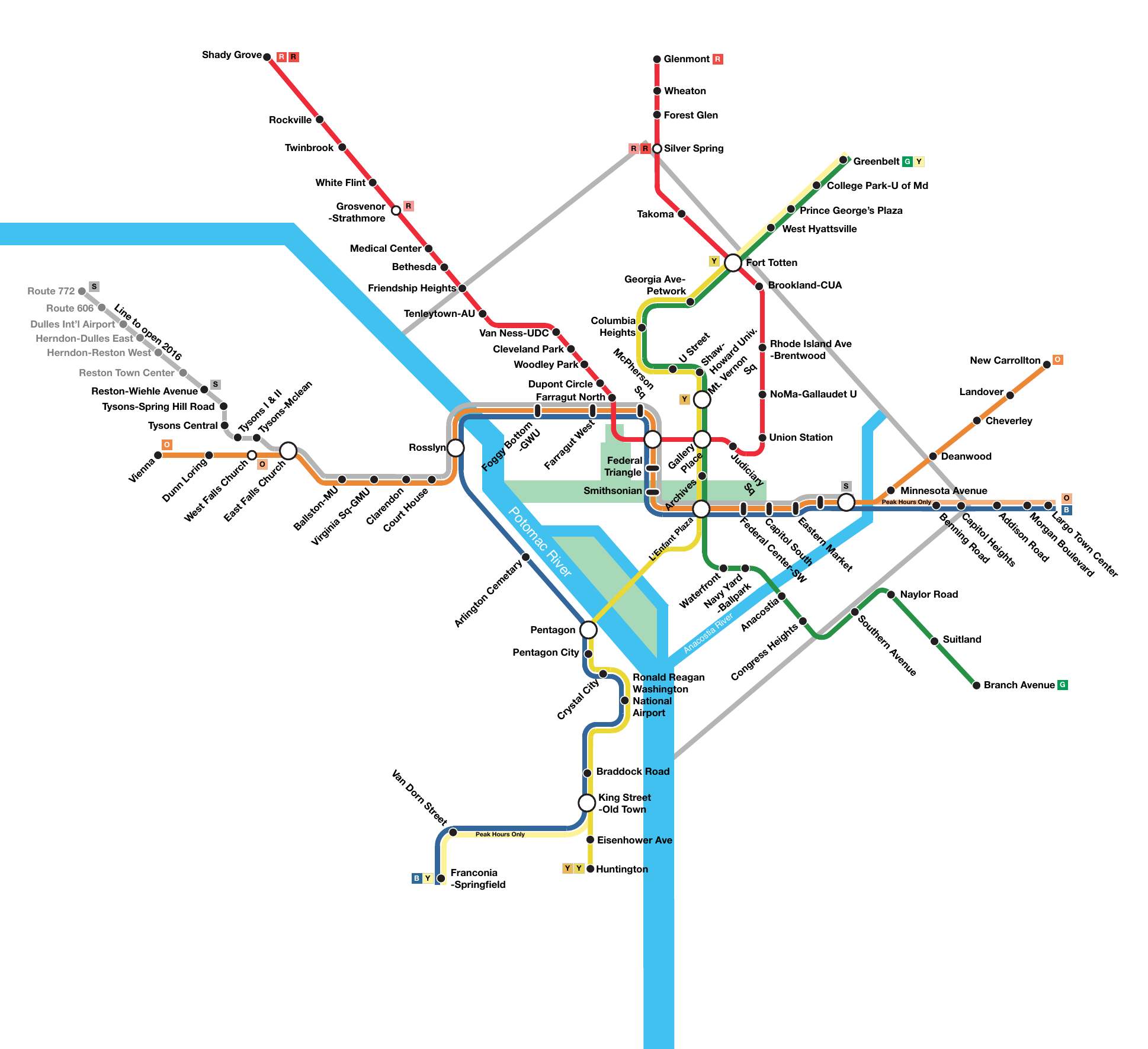}
    \caption{Map of the DC Metro network in 2012. Credit: VeggieGarden, originally posted to Wikipedia.}
    \label{fig:dcMetroMap}
\end{figure}

\section*{Appendix: Counterfactuals with Minimum Deviation Gains}

Tables \ref{table:CounterfactualsThresholdUniform} and \ref{table:CounterfactualsThresholdProportional} summarize the total outcomes of the DC-pricing counterfactuals when passengers only deviate from buying the full ticket if doing so reduces their ticket payments by at least the indicated amount.

\begin{table}[p]
  \centering
  \begin{tabular}{l|c|c|c|c}
   \textbf{Minimum gain} & \textbf{No ticket} & \textbf{Partial ticket} & \textbf{Full ticket} & \textbf{Revenue loss}\\
   \hline
  US\$0.50 & 68.11\% & 2.66\% & 29.22\% & 61.58\%\\
  US\$1.00 & 68.12\% & 0.43\% & 31.45\% & 61.12\%\\
  US\$2.00 & 68.12\% & 0.03\% & 31.85\% & 60.97\%\\
  US\$3.00 & 23.14\% & 0.00\% & 76.86\% & 29.08\%\\
  US\$4.00 & 9.89\% & 0.00\% & 90.11\% & 15.06\%\\
  US\$5.00 & 4.52\% & 0.00\% & 95.48\% & 7.71\%\\
  \end{tabular}
  \caption{Total counterfactual outcomes under DC pricing and random inspection with uniform monitoring, when deviations are only taken if ticket-payment savings reach the indicated amount. Revenue loss is measured relative to the benchmark in which all passengers pay the full DC fare.}
  \label{table:CounterfactualsThresholdUniform}
\end{table}

\begin{table}[p]
  \centering
  \begin{tabular}{l|c|c|c|c}
   \textbf{Minimum gain} & \textbf{No ticket} & \textbf{Partial ticket} & \textbf{Full ticket} & \textbf{Revenue loss}\\
   \hline
  US\$0.50 & 50.91\% & 19.98\% & 29.11\% & 50.85\%\\
  US\$1.00 & 50.91\% & 11.42\% & 37.67\% & 49.02\%\\
  US\$2.00 & 50.91\% & 3.61\% & 45.48\% & 45.77\%\\
   US\$3.00 & 13.48\% & 0.18\% & 86.33\% & 16.90\%\\
   US\$4.00 & 5.49\% & 0.00\% & 94.51\% & 8.17\%\\
   US\$5.00 & 2.13\% & 0.00\% & 97.87\% & 3.59\%\\
  \end{tabular}
  \caption{Total counterfactual outcomes under DC pricing and random inspection with proportional monitoring, when deviations are only taken if ticket-payment savings reach the indicated amount. Revenue loss is measured relative to the benchmark in which all passengers pay the full DC fare.}
  \label{table:CounterfactualsThresholdProportional}
\end{table}

\end{document}